\documentclass[preprint,amsmath, amssymb, aps, prl, showpacs,superscriptaddress]{revtex4-1}

\usepackage[utf8]{inputenc}
\usepackage{hyperref}
\usepackage{todonotes}
\usepackage{physics}
\usepackage{mathtools}
\usepackage{bbm}
\usepackage{comment}
\usepackage{enumerate}
\usepackage{amsthm}

\newtheorem{theorem}{Theorem}
\newtheorem{lemma}{Lemma}

\begin{document}

\title{Emergent Gauge Symmetries and Quantum Operations}

\author{A.P. \surname{Balachandran}}
\email{balachandran38@gmail.com}
\affiliation{Department of Physics, Syracuse University, Syracuse, New York 13244-1130, USA}

\author{I.M. Burbano}
\email[]{im.burbano10@uniandes.edu.co}
\affiliation{Departamento de F\'{i}sica, Universidad de los Andes,  A.A. 4976-12340, Bogot\'a, Colombia}
\affiliation{ICTP South American Institute for Fundamental Research\\
Instituto de F\'isica Te\'orica UNESP - Universidade Estadual Paulista
Rua Dr. Bento T. Ferraz 271, 01140-070, S\~ao Paulo, SP, Brasil}

\author{A.F. Reyes-Lega}
\email[]{anreyes@uniandes.edu.co}
\affiliation{Departamento de F\'{i}sica, Universidad de los Andes,  A.A. 4976-12340, Bogot\'a, Colombia}

\author{S. Tabban}
\email[]{sm.tabban@uniandes.edu.co}
\affiliation{Departamento de F\'{i}sica, Universidad de los Andes,  A.A. 4976-12340, Bogot\'a, Colombia}

\begin{abstract}
The  algebraic approach to quantum physics  emphasizes the role played by the structure  of the algebra of observables and its relation to the space of states. An important feature of this point of view is that subsystems can be described by subalgebras, with partial trace being replaced by the more general notion of restriction to a subalgebra.
This, in turn, has recently led to applications to the study of entanglement in systems of identical particles. In the course of  those investigations on entanglement and particle identity, an emergent gauge symmetry has been found by Balachandran, de Queiroz and Vaidya. In this letter we establish a novel connection  between that gauge symmetry, entropy production and quantum operations. Thus, let \textbf{A} be a system described by a finite dimensional observable algebra and $\omega$ a mixed faithful  state. Using the Gelfand-Naimark-Segal (GNS) representation we construct a canonical purification of $\omega$, allowing us to embed \textbf{A} into a larger system \textbf{C}. Using Tomita-Takasaki theory, we
obtain a subsystem decomposition of \textbf{C} into subsystems  \textbf{A} and \textbf{B}, without making use of any tensor product structure. We identify a  group of transformations that acts as a gauge group on \textbf{A} while at the same time giving rise to entropy increasing quantum operations on \textbf{C}. We provide physical means to simulate this gauge symmetry/quantum operation duality.
\end{abstract}

\maketitle

\emph{Introduction.}---The language of operator algebras has been shown to reveal the fundamental mathematical structure of quantum physics \cite{Haag1996,Bratteli1997,Moretti2012,Landsman2017}. In it,
the emphasis of the theory is placed on the abstract structures underlying the physical notions of \emph{observables} (described in terms of algebras) and \emph{states} (described in terms of positive linear functionals), as well as  on the relations between them.
 The understanding of quantum theory in such general terms has proven useful to the development of areas such as statistical mechanics \cite{Bratteli1997}, quantum field theory \cite{Haag1996}, and information theory \cite{OhyaPetz1993,Maassen2010}. Among the most important tools used in these developments are the Gelfand-Naimark-Segal (GNS) construction and the \emph{modular theory} of Tomita and Takesaki.
 The former provides a way of constructing a Hilbert space representation of the algebra of observables, once a state has been chosen. In the algebraic approach to quantum physics, Hilbert space is not considered as a priori given, and is thus  an emergent concept. This has important consequences. In statistical mechanics, for example, a better understanding of symmetry breaking is gained from the study of inequivalent Hilbert space representations~\cite{Landsman2017}. The foundational results of Doplicher, Haag and Roberts  on the structure of superselection sectors in quantum field theory~\cite{Doplicher1969,Doplicher1969a,Doplicher1971,Doplicher1974}  also rely on the same general principle.
 On the other hand, modular theory is a mathematical tool for probing the structure of the so called von Neumann algebras. Physically, it has found applications relating equilibrium states to dynamics \cite{Bratteli1997}, localization in algebraic quantum field theory \cite{Borchers2000}, and the GNS approach to thermo-field theory \cite{Ojima1981}. There are, moreover, hints suggesting that it could provide a solution to the problem of time in quantum gravity~\cite{Connes1994}.

The algebraic approach was used in \cite{Benatti2012,Balachandran2013, Balachandran2013a,Benatti2014a} to formulate a theory of entanglement based on the transparent idea of restrictions to subsystems rather than relying on partial traces. This had the advantage of being immediately applicable to problems for which the partial trace  approach failed, such as in the study of entanglement properties of identical particles. The assignment of an entropy to a given state was done there  by constructing a density operator on the GNS Hilbert space. As first pointed out by Sorkin~\cite{Balachandran2013a},  the decomposition of this space into irreducible subrepresentations is not  unique, so that different decompositions may yield density operators with different entropies. This was explored in \cite{Balachandran2013e,Balachandran2013d}, where such an ambiguity was traced back to an action of a unitary group giving rise to an \emph{emergent gauge symmetry}.

In this letter we establish a fundamental bridge between the theory of gauge fields and quantum information theory
through a novel implementation of this emergent gauge symmetry in terms of quantum operations.
Making use of the canonical purification afforded by the GNS construction, we  will enlarge the algebra of observables, obtaining an emergent environment on which the gauge group acts non-trivially. Making use of modular theory,  we will be able to reinterpret this gauge symmetry in terms of a concrete quantum operation in the purification environment.
Due to the robustness of the modular theory,  our results are valid for a wide class of systems, including any quantum system with a finite number of degrees of freedom, quantum theory on multiply connected configuration spaces and quantum fields on different types of backgrounds. For the sake of clarity here we will restrict to the case of finite dimensional algebras, as these are simple enough from a computational point of view  to allow us to describe  the relationship between gauge symmetries and quantum operations in an explicit fashion. However, the facts we establish are derived using tools that are immediately applicable to more general situations.

\emph{The algebraic approach.}---
In the algebraic approach to quantum physics, the description of a quantum system is given in terms of an abstract algebra $\mathcal{A}$. Basic properties, like commutation relations, are encoded in the structure of the algebra. States, in this approach, are described in terms of linear functionals $\omega: \mathcal{A}\rightarrow \mathbb C$.
Important defining properties of such algebras and states can be justified from a physical point of view, e.g.,  by looking at the basic example where the algebra is the space of bounded linear operators on a Hilbert space, $\mathcal A=\mathcal L(\mathcal H)$. This leads, in the abstract setting, to a formulation
of the theory in  terms of  a $C^*$-algebra, i.e.,  a triple $(\mathcal A, \|\cdot\|,*)$ consisting of a complex algebra $\mathcal A$ with an involution $``*"$ (playing the role of the adjoint) and a norm $\|\cdot\|$ fulfilling the $C^*$-property $\|aa^*\|=\|a\|^2$. A state $\omega$ on $\mathcal A$ is then defined as a normalized positive linear functional $\omega: \mathcal A \rightarrow \mathbb C$. If $a=a^*\in \mathcal A$ is an observable, $\omega(a)$ is the expectation value of $a$ in the state $\omega$.
A readable exposition of this point of view can be found in
\cite{Balachandran2013,Balachandran2013a,Reyes-Lega2016}.

From now on, let $\mathcal A$ be a finite-dimensional $C^*$-algebra with a faithful state $\omega$, that is, $\omega(a^*a)=0$ implies $a=0$ for all $a\in \mathcal A$.
In \emph{this} particular case the GNS construction allows one to construct a Hilbert space representation of the algebra in a very simple way. Note that $\mathcal A$ has a vector space structure.  Whenever we want to consider an element $a\in \mathcal A$ as a \emph{vector},  we will denote it by the ket $|a\rangle$. We can equip $\mathcal A$  with an inner product given by $\langle a |b\rangle = \omega(a^*b)$. The GNS representation is a representation of the algebra $\mathcal A$ on the GNS Hilbert space $\mathcal H_\omega\equiv (\mathcal A, \langle \cdot |\cdot\rangle )$. It is a $*$-homomorphism $\pi_\omega: \mathcal A  \rightarrow  \mathcal L(\mathcal H_\omega)$ which is naturally induced by the multiplication in $\mathcal A$, as follows.  Each element $a\in\mathcal A$ is mapped to an operator $\pi_\omega(a)\in\mathcal L(\mathcal H_\omega)$, defined by
$\pi_\omega(a)|b\rangle := |ab\rangle$.
The structure theorem of finite dimensional $C^*$-algebras~\cite{Takesaki2002} guarantees the existence of  unique $n_1,\dots,n_N\in\mathbb{N}^+$ such that $\mathcal{A}=\bigoplus_{r=1}^N\mathcal{A}_r$, where $\mathcal{A}_r\cong M_{n_r}(\mathbb{C})$, $1\leq r\leq N$.  We then have $\dim\mathcal{A}=\sum_{r=1}^N n_r^2$. Let $\mathbbm{1}_{\mathcal{A}_r}\in\mathcal{A}$ be the orthogonal projection onto $\mathcal{A}_r$ and $P^r:=\pi_\omega (\mathbbm{1}_{\mathcal{A}_r})$. This induces a decomposition of the GNS space into  subrepresentations $\mathcal{H}_\omega =\bigoplus_{r=1}^N\mathcal{H} ^r$, where $\mathcal{H} ^r:=P^r\mathcal{H}_\omega$.

In order to focus on the physical features of our proposal, we will denote
the system with observable algebra $\mathcal A$ as system $\mathbf{A}$
and will regard the GNS space $\mathcal H_\omega$ as describing an emergent system, denoted system $\mathbf{C}$.  Its observable algebra will be defined as
$\mathcal F:= \mathcal L({\mathcal{H}_\omega})$.

One can naturally identify $\mathcal{A}$ as a hermitian subalgebra of $\mathcal{F}$. Indeed, for all $a\in\mathcal{A}$ we have $\pi_\omega(a)\in\mathcal{F}$.
As $\pi_\omega$ is a $*$-representation,   we  have that the $\pi_\omega(a)$'s satisfy the same algebraic relations as the $a$'s. Moreover, $\omega$ being faithful implies that $\pi_\omega(a)\neq\pi_\omega(b)$ if $a\neq b$.  This entitles us to identify $\mathcal{A}$ and $\pi_\omega(\mathcal{A})\subseteq\mathcal{F}$. Physically, since systems are completely described by their observable algebras, this means that $\mathbf{A}$ can be identified as a subsystem of the emergent system $\mathbf{C}$. The embedding $\mathbf{A} \hookrightarrow \mathbf{C}$ can be further interpreted as a purification. Indeed, consider the  state vector $|\Omega \rangle:=|\mathbbm{1}_{\mathcal A}\rangle\in\mathcal{H}_\omega$.
We then have $\pi_\omega(a)|\Omega\rangle=|a\rangle$, so that
\begin{equation}\label{eq:GNS-purification}
\langle \Omega | \pi_\omega (a) |\Omega\rangle = \langle \mathbbm{1}_{\mathcal A} |a\rangle =\omega (a).
\end{equation}
This identity means that when we restrict the \emph{pure} state $|\Omega\rangle$ from the full algebra $\mathcal F$ to the subalgebra $\pi_\omega(\mathcal A)\cong \mathcal A$ we obtain the original state $\omega$ on $\mathcal A$. As restriction of a state to a subalgebra is a generalization of partial trace~\cite{Balachandran2013}, we conclude that $|\Omega\rangle$ provides a purification of $\omega$. This begs the question, which algebra can be called the subsystem $\mathbf{B}$ of $\mathbf{C}$ complementary to $\mathbf{A}$? Notice in particular that nowhere in this construction are we relating the notion of a subsystem decomposition to an eventual tensor product structure of system $\mathbf{C}$. The answer to the above question will be given by the modular theory to be described below.

\emph{Tomita-Takesaki (modular) theory.}---Two important properties obeyed by the vector state $|\Omega\rangle$ defined above are: (i) it is a \emph{cyclic vector} for $\mathcal H_\omega$, meaning that the span of all vectors of the form $\pi_\omega(a)|\Omega\rangle$ is $\mathcal H_\omega$ (upon completion, this is a general feature of the GNS construction, valid for general choices of $\mathcal A$ and $\omega$) and (ii) it is a \emph{separating vector} with respect to $\pi_\omega(\mathcal A)$, in the sense that $\pi_\omega(a)|\Omega\rangle=0$ implies $\pi_\omega(a)=0$. This is a direct consequence of our choice of the state $\omega$ as faithful.
The Tomita-Takesaki theory applies to any von Neumann algebra (a special type of $C^*$-algebra that includes finite matrix algebras~\cite{Takesaki2002}) acting as an algebra of operators on a Hilbert space which has a cyclic and separating vector. We can thus apply it to $\pi_\omega(\mathcal A)$.
Define an \emph{antilinear} operator $S$ on $\mathcal{H}_\omega$ by $S|a\rangle=|a^*\rangle$.
 Consider its polar decomposition $S=J\Delta^{1/2}$, where  $\Delta=S^*S$ is a positive hermitian operator  called the \emph{modular operator}. The antiunitary $J$ resulting from the polar decomposition is called the \emph{modular conjugation}, and it satisfies $J=J^*$ as well as $J^2=\textrm{id}$. The Tomita-Takesaki theorem~\cite{Haag1996} states that
\begin{align}
\label{eq:TT-theorem}
\begin{split}
    J\pi_\omega(\mathcal{A})J=&\,\pi_\omega(\mathcal{A})',\\
    \Delta^{it}\pi_\omega(\mathcal{A})\Delta^{-it}=&\,\pi_\omega(\mathcal{A}),
\end{split}
\end{align}
for all $t\in\mathbb{R}$. Here $\pi_\omega(\mathcal{A})'$ is the \emph{commutant} of $\pi_\omega (\mathcal{A})$, that is, the set of all operators $B\in\mathcal L(\mathcal{H}_\omega)$ such that $[B,\pi_\omega(a)]=0$ for all $a\in\mathcal{A}$.
The first identity in (\ref{eq:TT-theorem}) states that the commutant is completely determined by the modular conjugation $J$.  The second one corresponds to the fact that the modular operator $\Delta$ gives rise, in a canonical way, to a time evolution under which $\pi_\omega (\mathcal A)$ remains invariant.

As discussed above, system $\mathbf A$ is described by $\pi_\omega (\mathcal A)$ as a subsystem of $\mathbf C$. Hence,  the commutant $\pi_\omega(\mathcal{A})'$ (which is also a subsystem of $\mathbf C$) describes the system $\mathbf B$ complementary to $\mathbf A$.

\emph{Gauge symmetry.}---The decomposition $\bigoplus_{r=1}^N \mathcal H^r$ of $\mathcal H_\omega$ described above is unique in the sense that it is completely determined by the orthogonal projectors $P^r\in \pi_\omega (\mathcal A)$. Each $\mathcal H^r$ can be further decomposed into irreducible subrepresentations with multiplicity $n_r$, but for $n_r>1$ this decomposition is highly non-unique. This gives rise to a non-abelian gauge symmetry (which becomes abelian for $n_r=1$), as we now describe.

Notice that each component $\mathcal A_r$ of $\mathcal A$, being isomorphic to a simple matrix algebra, has a system of \emph{matrix units}. These are elements $e^{(r)}_{ij}\in \mathcal A_r$ such that: (i) $e^{(r)}_{ij} e^{(s)}_{k\ell}=\delta_{rs}\delta_{jk}e^{(r)}_{i\ell}$, (ii) $e^{(r)*}_{ij}=e^{(r)}_{ji}$ and (iii) $\sum_{i=1}^{n_r} e^{(r)}_{ii}=\mathbbm{1}_{\mathcal{A}_r}$.
Let $g$ be any element of the group $G\equiv U_{\mathcal A}$ of unitary elements of $\mathcal A$ and define  projectors $p_g^{(r,k)}:= g\, e^{(r)}_{kk}\,g^*$, as well as  $\mathcal  H^{(r,k)}_g:= P_g^{(r,k)}\mathcal H_\omega$, where
\begin{equation}
\label{eq:P_g^k}
 P_g^{(r,k)} :=J \pi_\omega(p_g^{(r,k)})J \in \mathbf B.
\end{equation}
Since $P^r\in \pi_\omega(\mathcal A)$, due to  (\ref{eq:TT-theorem}) the unitaries $U(g):=J\pi_\omega(g) J$ commute with the sum $\sum_{k=1}^{n_r}P_g^{(r,k)}=P^r$, and we obtain
$\mathcal H^r= \bigoplus_{k=1}^{n_r}\mathcal H^{(r,k)}_g$. Furthermore, the $g$-dependent decompositions
$\mathcal H_\omega =\bigoplus_{r=1}^N \bigoplus_{k=1}^{n_r}\mathcal H^{(r,k)}_g$
provide unitarily equivalent decompositions into irreducibles of the GNS representation $\pi_\omega$.
This  can be explicitly  shown by noticing that the
  vectors $\lbrace U(g)|e^{(r)}_{ik}\rangle\rbrace_{i=1,\ldots,n_r}$
span $\mathcal H^{(r,k)}_g$. A short computation using the definition of $\pi_\omega$ then shows that they transform under the action of $\pi_\omega (\mathcal A)$  in a  way that does not depend  on $g$, thus proving the statement.
The decomposition corresponding to $g=\mathbbm{1}_{\mathcal{A}}$ was used in~\cite{Balachandran2013,Balachandran2013a} in order to assign an entropy to the state $\omega$.  The equivalence of these decompositions then leads to $g$-dependent values of the entropy. This point was first raised in~\cite{Balachandran2013a}
and was extensively explored in \cite{Balachandran2013e,Balachandran2013d,Acharyya2014} to study the ensuing entropy ambiguities.
Here we have shown that the action of $G$ arises naturally in the context of modular theory.

There are several reasons that justify regarding $G$ as a \emph{gauge group} for $\mathcal A$. The fact that the operators $U(g)$ belong to system $\mathbf B$ means that their action will remain unnoticed, as far as system $\mathbf A$ is concerned. Furthermore, the examples discussed in~\cite{Balachandran2013d} show that for certain configuration spaces of molecular shapes, these transformations precisely correspond to actual gauge transformations on the fibre bundle where the molecular wave function is defined. But there is a further very compelling reason that appears as a result of theorem \ref{thm:1} below.
In order to appreciate its  meaning, we  emphasize  that in our current setting we have two algebras: (i) the algebra $\mathcal{F}$, which   is to be thought of as the algebra of fields in the sense of~\cite{Doplicher1969}, and (ii) the algebra $\mathcal{A}\cong\pi_\omega(\mathcal{A})$, that  plays the role of the observable algebra for system $\mathbf A$. Now,  in the theory of \emph{superselection sectors}, the gauge group plays a crucial role in selecting  the  observable algebra out of the field algebra~\cite{Doplicher1969,Doplicher1969a,Haag1996}. We claim
that in our case the gauge group coincides with  the group $G=U_{\mathcal A}$, acting on $\mathcal{H}_\omega$ via the representation $U(g)=J\pi_\omega(g)J$ for all $g\in G$.
First let us remark that
$U(G)$ is the set of unitary elements of $\pi_\omega(\mathcal{A})'$, i.e., $U(G)=U_{\pi_\omega(\mathcal{A})'}$. This follows from (\ref{eq:TT-theorem}) on using the facts that $J^2=\textrm{id}$ and that $\pi_\omega$ is faithful.
\begin{theorem}
\label{thm:1}
$\pi_\omega(\mathcal{A})=\mathcal{F}\cap U(G)'=U(G)'$
\end{theorem}
\begin{proof}---Note that trivially $U(G)'\subseteq\mathcal{F}$. Therefore $\mathcal{F}\cap U(G)'=U(G)'$. We thus only have to prove that $U(G)'=\pi_\omega(\mathcal{A})$.
Via the Tomita-Takesaki theorem, we have that $U(G)\subseteq\pi_\omega(\mathcal{A})'$. Since $\pi_\omega(\mathcal{A})$ is a von Neumann algebra, we have the defining property $\pi_\omega(\mathcal{A})=\pi_\omega(\mathcal{A})''$, from which $\pi_\omega(\mathcal{A})\subseteq U(G)'$ follows.
For the other inclusion note that proving $U(G)'\subseteq\pi_\omega(\mathcal{A})=\pi_\omega(\mathcal{A})''$ is equivalent to showing that for every $A\in U(G)'$ and $B\in \pi_\omega(\mathcal{A})'$ we have that $[A,B]=0$. As every element of a $C^*$-algebra is a linear combination of four unitary elements~\cite{Bratteli1997},
we have that $B=\sum_{i=1}^4 k_i U(g_i)$ for some $k_1,\dots,k_4\in\mathbb{C}$ and $g_1,\dots,g_4\in G$. Then, by definition of $U(G)'$ we have
$[A,B]=\sum_{i=1}^4 k_i[A,U(g_i)]=0$.
\end{proof}
The above result shows that $G$ is a gauge group in the sense that the observables $\mathcal{A}$ are precisely those operators in $\mathcal{F}$ which commute with the representation $U(G)$. This coincides with the characterization of the observable algebra introduced by Doplicher, Haag and Roberts as the gauge invariant part of the field algebra
(see, in particular, eqns. (1.9) and (3.3) in ~\cite{Doplicher1969}, as well as section I.3 in ~\cite{Connes1994a}).

\emph{Quantum operations from gauge symmetry.}---
By means  of the GNS construction we are modeling system $\mathbf A$   as a subsystem of $\mathbf C$. It makes sense, therefore, to consider the set of all states on $\mathbf C$ which, upon restriction to  $\mathcal A$, coincide with $\omega$. We will call such states \emph{extensions} of $\omega$.
In view of (\ref{eq:GNS-purification}), one such state is the canonical purification $\ket{\Omega}$ of $\omega$ afforded by the GNS construction. Extensions of $\omega$ display an interesting behavior with respect to a special class of quantum operations, that we now define.

 Let  $\lbrace \Lambda_k\rbrace_k$ be a family of operators such that (i) $\Lambda_k\in \pi_\omega(\mathcal A)'$ for all $k$ and (ii) $\sum_k \Lambda_k^*\Lambda_k =\mathbbm 1$, and define a quantum operation on system $\mathbf C$ through
$
\mathcal E_\Lambda (\rho):=\sum_k \Lambda_k\,\rho\,\Lambda_k^*.
$
Let now $\rho$ be a density operator on $\mathcal H_\omega$ representing an extension
of $\omega$. Then a short computation using (i), (ii) and the cyclicity of the trace shows that $\mathcal E_\Lambda (\rho)$ is  again an extension, i.e., $\Tr_{\mathcal H_\omega}(\mathcal E_\Lambda (\rho)\, \pi_\omega(a))= \Tr_{\mathcal H_\omega}(\rho\, \pi_\omega(a))=\omega (a)$ for all $a\in \mathcal A$.

Among this class of quantum operations we shall consider the projective measurements determined, for each $g\in G$, by the projectors $P_g^{(r,k)}$ defined in (\ref{eq:P_g^k}). They are directly related to the gauge symmetry, as discussed above. The quantum operation determined by the choice $\Lambda_{r,k} = P_g^{(r,k)}$ for the Kraus operators will be denoted by $\mathcal E_g$. Acting with this operation on $|\Omega\rangle$ we obtain the following family of density matrices on $\mathbf C$, parametrized by $G$:
\begin{equation}
\label{eq:rho_g}
\rho_g:= \mathcal E_{g}(|\Omega\rangle\langle \Omega|).
\end{equation}
The choice $g=\mathbbm{1}_{\mathcal A}$ gives a density operator $\mathcal E_{\mathbbm{1}_{\mathcal A}}(|\Omega\rangle\langle \Omega|)\equiv \rho_1$.  In the case $N=1$, this is precisely the one that was  used in~\cite{Balachandran2013} in order to compute entanglement entropies arising from restrictions.
The density operator $\rho_1$ arises there directly from the decomposition of $\mathcal H_\omega$ into irreducible subspaces. But, due to the gauge symmetry, the decomposition 
 induced by (\ref{eq:P_g^k}) for any $g\in G$ is physically equivalent to the former, as far as system $\mathbf A$ is concerned. The corresponding decomposition of $|\Omega\rangle\langle \Omega|$ then gives $\rho_g$ as result. Herein lies the connection between gauge transformations and quantum operations.

Each of these density matrices has a different $g$-dependent (von Neumann) entropy. Indeed, one can easily  check that the non trivial eigenvectors of $\rho_g$ are $P_g^{(r,k)}|\Omega\rangle$, with eigenvalues
 $       \lambda_{r,k}(g):=\|P_g^{(r,k)}\ket{\Omega}\|^2.$
Therefore, we have
\begin{equation}\label{eq:entropy}
    S(\rho_g)=-\sum_{r=1}^N\sum_{k=1}^{n_r}\lambda_{r,k}(g)\log\lambda_{r,k}(g).
\end{equation}
The entropy ambiguity is then parametrized by the group $G$.

In what follows it will be essential to distinguish the Hilbert space trace (denoted ``$\Tr_{\mathcal H}$'') from the trace of an algebra $\mathcal A$, which will be written in lowercase, as follows: ``$\tr_{\mathcal A}$''.
We will
use the fact~\cite{OhyaPetz1993} that, for each given state $\omega$, there is a unique positive element $R\in \mathcal A$ such that $\omega (a)=\tr_{\mathcal A}(Ra)$. Notice in particular that
\begin{equation}
\hspace{-0.2cm}\tr_\mathcal{A}(a)=\sum_{r=1}^N\frac{1}{n_r}\Tr_{\mathcal H^r}(\pi_\omega(a))
                  =\tr_{\pi_\omega(\mathcal{A})’}(J\pi_\omega(a^*)J).
\end{equation}
\begin{lemma}
\label{lem:1}
Let $\rho$ be any density operator on system $\mathbf C$ such that $\rho|_{\mathbf A}\equiv \omega$ and $J \rho J=\rho$.
Then, the unique density operator on $\pi_\omega(\mathcal A)'$ implementing
$\mathcal E_g(\rho)|_{\mathbf B}$ is $\mathcal E_g(J \pi_\omega (R)J)$.
\end{lemma}
\begin{proof}---The claim is that the density operator $\mathcal E_g(\rho)\in\mathcal L(\mathcal H_\omega)$ (which defines a state on $\mathbf C$), when restricted to $\pi_\omega(\mathcal A)'$ (system $\mathbf B$), coincides with $\mathcal E_g(J \pi_\omega (R)J)$, in the sense that for all $B\in \pi_\omega(\mathcal A)'$ the following identity holds:
\begin{equation}
\label{eq:identity-thm2}
\Tr_{\mathcal H_\omega}(\mathcal E_g(\rho)B)=\tr_{\pi({\mathcal A})'}\left(\mathcal E_g(J\pi_\omega (R)J) B\right).
\end{equation}
By modular theory,  $B=J\pi_\omega(a)J$, for some $a\in \mathcal A$. Using (\ref{eq:P_g^k}), $J^2=\textrm{id}$ and the fact that  $\Tr (ST) = \Tr(S^*T^*)$ whenever both $S$ and $T$ are antilinear operators\cite{Uhlmann2016}, we obtain
$\Tr_{\mathcal H_\omega}(\mathcal E_g(\rho)J\pi_\omega(a)J)=\sum_{r,k}\omega (p_g^{(r,k)} a^*p_g^{(r,k)})$. To compute the RHS of (\ref{eq:identity-thm2}) we note that
the linear map $a\mapsto \tr_{\pi({\mathcal A})'}(J\pi_\omega (a^*)J)$ defines a \emph{trace} on $\mathcal A$. It then follows, from uniqueness of the trace for finite algebras, that
$
\tr_{\pi({\mathcal A})'}(J\pi_\omega (a^*)J)\equiv \tr_{\mathcal A}(a).
$
Using this result, together with $\omega(a)=\tr_{\mathcal A}(Ra)$, (\ref{eq:P_g^k}) and $J^2=\mathrm{id}$, we obtain  (\ref{eq:identity-thm2}).
\end{proof}
\begin{theorem}
\label{thm:2} $S(\rho_g)\ge S(\rho_1)$ for all $g\in U_{\mathcal A}$.
\end{theorem}
\begin{proof}---Since $J|\Omega\rangle=|\Omega\rangle$, it follows from lemma \ref{lem:1} and the definition in (\ref{eq:rho_g}) that $S(\rho_g|_{\mathbf B})=S(\mathcal E_g(J\pi_\omega (R)J))$. Since $\mathcal E_g$ is a projective measurement, it follows that $S(\rho_g|_{\mathbf B})\ge S(J\pi_\omega (R)J)$ $=S(R)=S(\rho_1)$. Moreover, using the projectors  furnished by the spectral theorem when applied to $R$ to construct the matrix units $e^{(r)}_{ij}$, one can check that $\rho_1$ and $R$ have exactly the same  entropy.
The result then follows from an explicit computation, using again the matrix units $e_{ij}^{(r)}$, which shows that  restriction of  $\rho_g$ to $\mathbf B$ does not change entropy, i.e., $S(\rho_g|_{\mathbf B})=S(\rho_g)$.
\end{proof}
As a consequence of this result we obtain the identity
$S(\rho_g|_{\mathbf{A}}) =
S(\rho_1)$, as well as $S(\rho_g|_{\mathbf{B}})=
S(\rho_g)$. We therefore see that the whole ambiguity in the entropy is carried by system $\mathbf B$.
The transformations induced by $g\in G$ are gauge transformations for system $\mathbf A$ and, from the point of view of $\mathbf A$, do not change the entropy. This provides further support to the idea proposed in \cite{Balachandran2013a} of using the GNS approach to study entanglement in systems of identical particles.
Notice also that the $g$-induced increase in entropy can be expressed as a relative entropy:
\begin{equation}
\Delta S \equiv S(\rho_g)-S(\rho_1)= S(\mathcal E_g(J\pi_\omega(R)J)||J\pi_\omega(R)J).
\end{equation}
\emph{Example: Bipartite system.}---
Let  $\mathcal A=M_n(\mathbb C)$, with matrix units $e_{ij}$. Put $\omega (a) =\tr (R a)$, $R=\sum_i \lambda_i e_{ii}$,  with $R$ invertible. The vectors $|\hat e_{ij}\rangle :=(\lambda_j)^{-1/2}|e_{ij}\rangle$ provide an orthonormal basis  for $\mathcal H_\omega$. It follows that the  linear map
$\Phi: \mathcal H_\omega\rightarrow \mathbb C^n \otimes \mathbb C^n$ defined  through $\Phi(|\hat e_{ij}\rangle):=|i\rangle\otimes |j\rangle$ is a Hilbert space isomorphism. Therefore, to any (anti-)linear  operator $T\in \mathcal L(\mathcal H_\omega)$ there is a corresponding operator $\widetilde T:=\Phi T \Phi^{-1}\in \mathcal L(\mathbb C^n \otimes \mathbb C^n)$.
Explicit computation then leads to $\widetilde J  \ket{i} \otimes \ket{j} = \ket{j} \otimes \ket{i}$, as well as to $\widetilde\pi_\omega(a)=a\otimes \mathbbm 1_n$
and $\widetilde J\widetilde\pi_\omega(a)\widetilde J=\mathbbm 1_n \otimes \bar{a}$,  $a\in \mathcal A$, where $\bar a$ is the complex conjugate of $a$.
We thus see that in this case the notion of complementary subsystems described above in terms of modular theory reduces to the usual  tensor product decomposition of a bipartite system.
Under the isomorphism, the  algebra of system $\mathbf{C}$ is $\Phi\mathcal{F}\Phi^{-1}=\mathcal{A} \otimes \mathcal{A}$ while that of  subsystem $\mathbf{A}$ is $\mathcal{A} \otimes\mathbbm{1}_n$. One then identifies  system $\mathbf{B}$ with the algebra $\mathbbm{1}_n\otimes \mathcal{A}$, such that $\mathbf{C}$ is precisely the composite system of $\mathbf{A}$ and $\mathbf{B}$.
The power of modular theory lies in the fact that it will produce a bipartition in more general situations, where tensor products might not be suitable (or even unavailable~\cite{Schroer2010}). Let now $g\in G\equiv U(n)$ and define $\lambda_k(g):=\sum_i \lambda_i|g_{ik}|^2$. For $\rho_g$ defined as in (\ref{eq:rho_g}) we find
$\widetilde \rho_g|_{\textbf B}=\sum_k \lambda_k(g)\,  \bar g|k\rangle\langle k| \bar{g}^*$, as well as $\widetilde{\mathcal E}_g(\widetilde J\widetilde \pi_\omega(R)\widetilde J)=\mathbbm 1_n\otimes (\widetilde \rho_g|_{\textbf B})$, in accordance with the statement of lemma \ref{lem:1}. Finally, the spectrum of $\rho_g$ is given by $\lbrace \lambda_k(g) \rbrace_{k=1,\ldots,n}$, whereas the spectrum of $\rho_1$ coincides (up to irrelevant zero eigenvalues) with that of $R$, in accordance with the statement of theorem \ref{thm:2}.

\emph{Conclusions}---We have obtained a novel duality relation between gauge symmetries and quantum operations which has its roots in the modular theory of Tomita-Takesaki. The properties of this emergent gauge symmetry have been worked out in full generality in the finite dimensional case. Our main results are: (i) the construction of a canonical embedding and purification of a quantum system by means of the GNS construction and the identification of a subsystem decomposition using modular theory, (ii) the identification of a gauge symmetry in the sense of Doplicher-Haag-Roberts and  (iii) the construction of a family of entropy-increasing quantum operations induced by gauge transformations, which nevertheless leave invariant the original system.
It would be interesting to explore the consequences of our results
for matrix models of gauge theories and, in particular, to the idea that color in QCD is mixed~\cite{Balachandran2015a, Balachandran2015b}. Our proposal can also be implemented in more general situations, such as systems arising from quantization of homogeneous spaces~\cite{Landsman1990}. Previous work on anomalies from a Hamiltonian point of view~\cite{Esteve1986,Gracia-Bondia1994,Balachandran2012} strongly suggests that there should be a connection between anomalies and the type of quantum operations considered here. We hope to return to this problem in the near future.

The authors would like to thank Aleksandr Pinzul for discussions that led to a significant improvement of this work. Financial support from Universidad de los Andes through projects INV-2017-51-1444,   INV-2017-26-1094 and INV-2018-34-1295  is gratefully acknowledged. I.M. Burbano thanks  ICTP-SAIFR, ICTP-Trieste and FAPESP grant 2016/03143-7 for partial financial support.
%

\end{document}